\documentclass[11pt,a4paper]{article}
\usepackage[margin=25mm]{geometry}
\usepackage[T1]{fontenc}
\usepackage{authblk}
\usepackage{lmodern}
\usepackage{amsmath,amssymb,mathtools}
\usepackage{amsthm}
\usepackage{thmtools}

\usepackage{microtype}
\usepackage[hidelinks]{hyperref}
\usepackage{enumitem}
\usepackage[authoryear]{natbib}
\setlist[itemize]{nosep,leftmargin=18pt}

\newtheorem{lemma}{Lemma}
\newtheorem{remark}{Remark}
\newtheorem{corollary}{Corollary}
\newcommand{\K}{\mathbb K}

\newcommand{\MM}[1]{\langle #1\rangle}

\title{A lower bound for \texorpdfstring{$\langle 3,2,m \rangle$}{(3,2,m)} matrix multiplication}

\author{Askar Tsyganov\footnote{Corresponding author: \texttt{atsyganov@hse.ru}}}
\author{Uliana Parkina}
\author{Sergey Samsonov}
\author{Maxim Rakhuba}

\affil{HSE University}

\date{}

\begin{document}
\maketitle
\begin{abstract}
We prove that, over any field, the bilinear complexity of multiplying a $3\times 2$ matrix by a $2\times m$ matrix is strictly greater than $24m/5$. In particular, every exact bilinear algorithm for multiplying a $3\times 2$ matrix by a $2\times 5$ matrix requires at least $25$ multiplications. Together with the Hopcroft-Kerr upper bound, this proves that the $\langle 3,2,5\rangle$ matrix multiplication tensor has rank exactly $25$. The proof has been formally verified in Lean~4, with the formalization available at \url{https://github.com/fallnlove/mm325_proof}.
\end{abstract}

\section{Introduction} 

Let $\K$ be an arbitrary field, and let $m \geq 1$ be an integer.
We study the bilinear complexity of multiplying a $3\times 2$ matrix
by a $2\times m$ matrix. A bilinear algorithm of length $r$ for the
corresponding matrix multiplication tensor $\MM{3,2,m}$ is an identity
\begin{equation}\label{eq:algorithm}
 AB=\sum_{i=1}^{r}\alpha_i(A)\beta_i(B)C_i,
 \qquad
 \forall A\in\K^{3\times2},\quad
 \forall B\in\K^{2\times m},
\end{equation}
where $\alpha_i:\K^{3\times2}\to\K$ and
$\beta_i:\K^{2\times m}\to\K$ are $\K$-linear maps, and
$C_i\in\K^{3\times m}$ are fixed matrices for $1\leq i\leq r$.
The tensor rank $R_{\K}(\MM{3,2,m})$, equivalently the bilinear
complexity of this multiplication, is the minimum length of such
an algorithm.

Several lower bounds are known for this family of tensors.
Over an arbitrary field $\K$, the general lower bound of
\citet{blaser1999lower}, together with the invariance of matrix
multiplication tensor rank under permutations of the dimensions,
gives
\[
 R_{\K}(\MM{3,2,5})
 =
 R_{\K}(\MM{2,3,5})
 \geq 22.
\]
Over $\mathbb{F}_2$, \citet{ja1985improved} proved the bound
$
 R_{\mathbb{F}_2}(\MM{3,2,m})
 \geq \left\lceil \frac{105m}{23} \right\rceil,$
which yields a lower bound of $23$ when $m=5$.
More recently, building on the work of \citet{shitov2026},
\citet{yang2026newlowerboundstensor} established the bounds
\[
 R_{\K}(\MM{n,2,m}) \geq (n+2)m
 \quad \text{for } n\geq 4, \quad\text{ and } \quad R_{\K}(\MM{3,2,m})
 \geq \left\lceil \frac{24m}{5} \right\rceil,
\]
both over arbitrary fields.
For $m=5$, the latter gives
$R_{\K}(\MM{3,2,5})\geq 24$.

In this paper, we refine the argument of
\citet{yang2026newlowerboundstensor} to exclude equality in the bound
$R_{\K}(\MM{3,2,m})\geq 24m/5$. The main idea of the proof was obtained by OpenAI's \texttt{gpt-6-astra}.

\begin{restatable}{theorem}{mainthm}
\label{thm:main}
Let $\K$ be a field, and let $m \geq 1$ be an integer. Then
\[
 R_{\K}(\MM{3,2,m})>\frac{24m}{5}.
\]
\end{restatable}

Since tensor rank is integer-valued, Theorem~\ref{thm:main}
improves \citep{yang2026newlowerboundstensor} lower bound by one whenever $m$ is divisible by $5$.
In particular, it gives $R_{\K}(\MM{3,2,5})\geq 25$.
Together with the matching upper bound of
\citet{hopcroft1971minimizing}, this determines the exact bilinear
complexity of multiplying a $3\times2$ matrix by a $2\times5$ matrix.

\begin{corollary}
Let $\K$ be a field. Then
\[
 R_{\K}(\MM{3,2,5}) = 25.
\]
\end{corollary}

\paragraph{Notation.}
We write $E_{ij}$ for the matrix with a $1$ in position $(i,j)$
and zeros elsewhere, with dimensions determined by context.
The symbol $\mathbb{F}_2$ denotes the field with two elements.

\section{Preliminary lemmas}

\begin{lemma}[Normalization]\label{lem:normalization}
Let $\K$ be an infinite field and let $m,r\in\mathbb Z_{>0}$.
Assume that there exists a bilinear algorithm of length $r$
that computes $AB$ exactly for every $A\in\K^{3\times2}$ and
$B\in\K^{2\times m}$, with all first-input linear forms nonzero. Then there exist
$\K$-linear forms $\alpha_i:\K^{3\times2}\to\K$,
$\beta_i:\K^{2\times m}\to\K$, and matrices
$C_i\in\K^{3\times m}$ ($i \in \{1, ..., r\}$), satisfying
\eqref{eq:algorithm} and
\[
 \forall i\in\{1,\ldots,r\},\qquad \alpha_i(E_{11})=1.
\]
With $s_i=\alpha_i(E_{21})$ and $t_i=\alpha_i(E_{31})$, these
forms also satisfy
\begin{equation}
\label{eq:dim_span}
 \forall(s,t)\in\K^2,\qquad
 \dim_{\K}\operatorname{span}_{\K}
 \bigl\{\alpha_i\mid 1\leq i\leq r,\ s_i=s,\ t_i=t\bigr\}
 \leq2.
\end{equation}
The span is taken in the dual space $(\K^{3\times2})^*$.
\end{lemma}

\begin{remark}
    A similar result was obtained by \citet{yang2026newlowerboundstensor}, but the dimension inequality~\eqref{eq:dim_span} is not proved there. For clarity, we reproduce the relevant part of the argument here.
\end{remark}

\begin{proof}
Write the columns of $A$ as $x,y\in\K^3$. Each original form has
unique coefficient vectors $a_i,b_i\in\K^3$ such that
\[
 \alpha_i([x\ y])=a_i^{\top}x+b_i^{\top}y.
\]
Choose a scalar $\tau$ with two properties: every $a_i+\tau b_i$ is
nonzero, and two such vectors are proportional only if
$a_i+t b_i$ and $a_j+t b_j$ are linearly dependent for every $t$.
This choice is possible for an elementary reason. Each vanishing
vector excludes at most one value of $\tau$. The $2\times2$ minors
of the matrix with columns $a_i+t b_i$ and $a_j+t b_j$ are
polynomials of degree at most two. Avoid the roots of every minor
that is not the zero polynomial. There are only finitely many
forbidden values.

Partition $\{1,\ldots,r\}$ into the equivalence classes of
\[
 i\sim j\quad\Longleftrightarrow\quad
 \exists c\in\K\setminus\{0\}:\quad
 a_i+\tau b_i=c(a_j+\tau b_j).
\]
Fix an equivalence class. Suppose that $a_i,b_i$ are linearly
independent for some index $i$ in this class, and let $j\sim i$.
By the choice of $\tau$,
\[
 \forall t\in\K,\qquad
 \operatorname{rank}
 \begin{bmatrix}a_i+t b_i & a_j+t b_j\end{bmatrix}
 \leq 1.
\]
Taking $t=0$ and $t=1$, respectively, gives
\[
 a_j=\lambda a_i,
 \qquad
 a_j+b_j=\mu(a_i+b_i)
\]
for some $\lambda,\mu\in\K$. Hence
\[
 b_j=(\mu-\lambda)a_i+\mu b_i.
\]
Therefore
\[
 \begin{bmatrix}a_i+t b_i & a_j+t b_j\end{bmatrix}
 =
 \begin{bmatrix}a_i & b_i\end{bmatrix}
 \begin{pmatrix}
  1 & (1-t)\lambda+t\mu\\
  t & t\mu
 \end{pmatrix}.
\]
Since $\begin{bmatrix}a_i & b_i\end{bmatrix}$ has full column
rank, the $2\times2$ matrix on the right has rank at most one.
Consequently,
\[
 \forall t\in\K,\qquad
 0=
 \det\begin{pmatrix}
  1 & (1-t)\lambda+t\mu\\
  t & t\mu
 \end{pmatrix}
 =t(1-t)(\mu-\lambda).
\]
Since $\K$ is infinite, we may choose $t\notin\{0,1\}$,
which yields $\mu=\lambda$. Thus
\[
 a_j=\lambda a_i,\qquad b_j=\lambda b_i,
 \qquad\alpha_j=\lambda\alpha_i.
\]
Hence the forms in this class span a space of dimension at most one.
If instead $a_i,b_i$ are dependent for every index in the group,
all these vectors lie on the same line: that line contains the
nonzero vector $a_i+\tau b_i$. For a fixed nonzero vector $a$ on
this line, every form in the group is a linear combination of
$a^{\top}x$ and $a^{\top}y$. This proves the claim.

Now make the change of columns $[x\ y]\mapsto[x\ y+\tau x]$.
The coefficient of the first column becomes $a_i+\tau b_i$.
Choose the first basis vector of $\K^3$ so that none of these
nonzero linear forms vanishes on it, and complete it to a basis.
To justify this choice, define
\[
 \begin{aligned}
 q_i(u)&=(a_i+\tau b_i)^{\top}(1,u,u^2)^{\top}\\
       &=(a_{i1}+\tau b_{i1})
        +(a_{i2}+\tau b_{i2})u+(a_{i3}+\tau b_{i3})u^2,
        \qquad 1\leq i\leq r.
 \end{aligned}
\]
Because $a_i+\tau b_i\neq0$, each $q_i\in\K[u]$ is a nonzero
polynomial of degree at most two. The union of their root sets
is finite, whereas $\K$ is infinite. Therefore
\[
 \exists u\in\K\ \forall i\in\{1,\ldots,r\},\qquad q_i(u)\neq0.
\]
Choose such a $u$ and take $(1,u,u^2)^{\top}$ as the first
basis vector. In the resulting coordinates, the matrix unit
\[
 E_{11}=\begin{pmatrix}1&0\\0&0\\0&0\end{pmatrix}
\]
has first column $(1,0,0)^{\top}$ and second column zero.
For the transformed forms, again denoted by $\alpha_i$, we have
\[
 \forall i\in\{1,\ldots,r\},\qquad \alpha_i(E_{11})=q_i(u)\neq0.
\]
Divide each $\alpha_i$ by this value, and multiply the corresponding
$\beta_i$ by the same value. The first-column coefficient vector
of the resulting form is $(1,s_i,t_i)^{\top}$.
Thus two pairs $(s_i,t_i)$ coincide exactly when their earlier
first-column coefficient vectors were proportional. The groups
are the ones already considered, and their span dimensions are
preserved by the invertible coordinate changes and nonzero rescalings.

For completeness, these coordinate changes preserve algorithm
length because they can be compensated in the second input and
in the output. For invertible $L\in\K^{3\times3}$ and $G\in\K^{2\times2}$,
\[
 (LAG)(G^{-1}B)=LAB.
\]
Apply the original algorithm to the two inputs on the left, then
multiply its output by $L^{-1}$. This proves the lemma for the
changed coordinates without changing the number of summands.
\end{proof}

\begin{lemma}[Dimension inequality and equality case]\label{lem:count}
Let $E$ be a finite-dimensional vector space over a field $\K$,
and let $J\subseteq H\subseteq E$ be linear subspaces with
$\dim J=m>0$ and $\dim H=h$.
Let $S,T,Z:E\to E$ and $X:J\to E$ be $\K$-linear maps such that
$ST=TS$, $SZ=ZS$, and $TZ=ZT$. Assume
\begin{equation}\label{eq:inside}
 \forall y\in J,\qquad Sy\in H,\quad Ty\in H,\quad Xy-Zy\in H,
\end{equation}
and
\begin{equation}\label{eq:independence}
 \forall a,u,v\in J,\qquad
 Za+SXu+TXv\in H\quad\Longrightarrow\quad a=u=v=0.
\end{equation}
Then $5h\geq9m$. If $5h=9m$, there exists a linear subspace
$F\subseteq J$ satisfying
\[
 F\neq\{0\},\qquad S(F)\subseteq F,\qquad T(F)\subseteq F.
\]
\end{lemma}

\begin{remark}
    A similar result was obtained by \citet{yang2026newlowerboundstensor}, but the case $5h = 9m$ is not addressed there. For clarity, we reproduce the relevant part of the argument here.
\end{remark}

\begin{proof}
Define
\[
 F=\{y\in J:Sy\in J\ \land\ Ty\in J\},\qquad f=\dim F.
\]
Consider the following homogeneous linear conditions on
$(u,v,c)\in J^3$:
\begin{equation}\label{eq:conditions}
 \begin{aligned}
 Su+Tv&\in J,\\
 Xu-Zu+Tc&\in J,\\
 Xv-Zv-Sc&\in J.
 \end{aligned}
\end{equation}
By~\eqref{eq:inside}, applying $S$ to the second expression and
$T$ to the third gives elements of $H$. Hence commutativity implies
\[
 \begin{aligned}
 &S(Xu-Zu+Tc)+T(Xv-Zv-Sc)\\
 &\hspace{2em}=SXu+TXv-Z(Su+Tv)\in H.
 \end{aligned}
\]
Since $Su+Tv\in J$, assumption~\eqref{eq:independence} applied
with $a=-(Su+Tv)$ yields $u=v=0$. The remaining conditions are
equivalent to $c\in F$. Conversely, each triple $(0,0,c)$ with
$c\in F$ satisfies~\eqref{eq:conditions}. Thus the solution
space is $\{0\}\times\{0\}\times F$, of codimension $3m-f$
in $J^3$.

Fix a basis of $H$ extending a basis of $J$.
By~\eqref{eq:inside}, each expression in~\eqref{eq:conditions}
belongs to $H$. Membership in $J$ is therefore equivalent to
the vanishing of its $h-m$ additional coordinates. Each condition
has rank at most $h-m$, and the rank of the combined system is
at most the sum of these ranks. Consequently,
\begin{equation}\label{eq:first-count}
 3m-f\leq3(h-m).
\end{equation}

The solution space of the first condition in $J\times J$
contains $F\times F$, so its codimension is at most $2m-2f$.
Applying this bound to the first condition and the preceding
bound to the other two gives
\begin{equation}\label{eq:second-count}
 3m-f\leq2(m-f)+2(h-m).
\end{equation}
Rearranging~\eqref{eq:first-count} and~\eqref{eq:second-count},
\[
 6m-3h\leq f\leq2h-3m.
\]
It follows that $9m\leq5h$.

Suppose $5h=9m$. Both endpoints coincide, giving $f=3m/5>0$.
The combined system has rank $3m-f=3(h-m)$, whereas the last
two conditions together have rank at most $2(h-m)$. The first
condition therefore has rank exactly $h-m=2(m-f)$. Its solution
space in $J\times J$ has dimension $2f$ and contains
$F\times F$, so it equals $F\times F$. Thus
\begin{equation}\label{eq:equality-pairs}
 \forall u,v\in J,\qquad
 Su+Tv\in J\quad\Longleftrightarrow\quad u\in F\ \land\ v\in F.
\end{equation}
Finally, if $y\in F$, then $Ty,-Sy\in J$ and
\[
 S(Ty)+T(-Sy)=0\in J.
\]
Equation~\eqref{eq:equality-pairs} gives $Ty,-Sy\in F$.
Since $F$ is a linear subspace, $S(F)\subseteq F$ and
$T(F)\subseteq F$.
\end{proof}

\begin{lemma}[Common eigenvector]\label{lem:support}
Let $\K$ be a field, let $r > 0$ be an integer, and let
$S,T\in\K^{r\times r}$ be diagonal matrices. Let
$F\subseteq\K^r$ be a linear subspace satisfying
\[
 F\neq\{0\},\qquad
 \forall y\in F,\quad Sy\in F\ \land\ Ty\in F.
\]
Then
\[
 \exists (s,t)\in\K^2, y\in F \setminus\{0\} :
 \qquad Sy=sy\ \land\ Ty=ty.
\]
\end{lemma}

\begin{proof}
Choose $y\in F\setminus\{0\}$ such that
\[
 \forall z\in F\setminus\{0\},\qquad
 \#\{i\mid y_i\neq0\}\leq\#\{i\mid z_i\neq0\}.
\]
All coordinate indices range over $\{1,\ldots,r\}$.
Since $y\neq0$, choose $j\in\{1,\ldots,r\}$ with $y_j\neq0$.
For an arbitrary $D\in\{S,T\}$, set $w=(D-D_{jj}I)y$.
The assumption $Dy\in F$ and linearity of $F$ give
\[
 w=Dy-D_{jj}y\in F.
\]
Since $D$ is diagonal,
\[
 \begin{aligned}
 &\forall i\in\{1,\ldots,r\},\qquad w_i=(D_{ii}-D_{jj})y_i,\\
 &\forall i\in\{1,\ldots,r\},\qquad y_i=0\ \Longrightarrow\ w_i=0,\\
 &w_j=(D_{jj}-D_{jj})y_j=0.
 \end{aligned}
\]
Consequently,
\[
 \{i\mid w_i\neq0\}\subseteq\{i\mid y_i\neq0\}\setminus\{j\},
 \qquad
 \#\{i\mid w_i\neq0\}<\#\{i\mid y_i\neq0\}.
\]
If $w\neq0$, this contradicts the defining inequality for $y$
with $z=w$. Thus
\[
 \forall D\in\{S,T\},\qquad (D-D_{jj}I)y=0.
\]
Taking $(s,t)=(S_{jj},T_{jj})$ yields $Sy=sy$ and $Ty=ty$.
\end{proof}

\section{The main theorem}

\mainthm*

\begin{proof}[Proof of Theorem~\ref{thm:main}]
\textbf{Normalization.}
Identity~\eqref{eq:algorithm} is determined by its values on pairs
of matrix units. Extending its coefficients to $\K(t)$ therefore
gives an algorithm of the same length over $\K(t)$, and hence
\[
 R_{\K}(\MM{3,2,m})\geq R_{\K(t)}(\MM{3,2,m}).
\]
It suffices to prove the bound when $\K$ is infinite.

Set $r=R_{\K}(\MM{3,2,m})$ and choose an algorithm of length $r$.
Minimality implies $\alpha_i\neq0$ for all $1\leq i\leq r$.
By Lemma~\ref{lem:normalization}, the algorithm may be chosen to
satisfy, with $s_i=\alpha_i(E_{21})$ and $t_i=\alpha_i(E_{31})$,
\[
 \forall i\in\{1,\ldots,r\},\qquad \alpha_i(E_{11})=1,
\]
and
\[
 \forall(s,t)\in\K^2,\qquad
 \dim_{\K}\operatorname{span}_{\K}
 \bigl\{\alpha_i\mid 1\leq i\leq r,\ s_i=s,\ t_i=t\bigr\}\leq2.
\]
Define
\[
 S=\operatorname{diag}(s_i),\qquad
 T=\operatorname{diag}(t_i),\qquad
 Z=\operatorname{diag}(\alpha_i(E_{12})).
\]
Then $ST=TS$, $SZ=ZS$, and $TZ=ZT$.

\textbf{Construction of the subspaces.}
Define the linear map
\[
 W:\K^r\longrightarrow\K^{3\times m},\qquad
 W(z) =\sum_{i=1}^{r}z_i C_i,
\]
and set $H=\ker W$.
By~\eqref{eq:algorithm}, $\operatorname{Im}W$ contains every
product $AB$, hence every $3\times m$ matrix unit. Therefore
$W$ is surjective, and the rank-nullity theorem gives
\begin{equation}\label{eq:H}
 \dim H=r-3m.
\end{equation}

Define the linear subspace
\begin{equation}\label{eq:J}
 J=\left\{
 \left(\beta_i\!\begin{pmatrix}0\\v\end{pmatrix}\right)_{i=1}^{r}
 :v\in\K^{1\times m}\right\}.
\end{equation}
Identity~\eqref{eq:algorithm}, evaluated at $A=E_{12}$, implies
\[
 \forall v\in\K^{1\times m},\qquad
 \left(\forall i\in\{1,\ldots,r\},\
 \beta_i\!\begin{pmatrix}0\\v\end{pmatrix}=0\right)
 \ \Longrightarrow\ v=0.
\]
Consequently, the linear parametrization of $J$ in~\eqref{eq:J}
is injective, so $\dim J=m$.
Using $\alpha_i(E_{11})=1$ in~\eqref{eq:algorithm} gives
\[
 \forall v\in\K^{1\times m},\qquad
 W\left(\left(\beta_i\!\begin{pmatrix}0\\v\end{pmatrix}\right)_i\right)
 =E_{11}\begin{pmatrix}0\\v\end{pmatrix}=0.
\]
Thus $J\subseteq H$. Define $X:J\to\K^r$ by
\begin{equation}\label{eq:X}
 X\left(\left(\beta_i\!\begin{pmatrix}0\\v\end{pmatrix}\right)_i\right)
 =\left(\beta_i\!\begin{pmatrix}v\\0\end{pmatrix}\right)_i,
 \qquad \forall v\in\K^{1\times m}.
\end{equation}
The injectivity of the parametrization implies that $X$ is
well-defined and linear.

\textbf{Application of the dimension inequality.}
For $v\in\K^{1\times m}$ and
$y=(\beta_i(\binom{0}{v}))_i\in J$, substitution of the matrix
units into~\eqref{eq:algorithm} yields
\begin{equation}\label{eq:table}
 \begin{aligned}
 W(y)&=W(Sy)=W(Ty)=0,\\
 W(Xy)&=W(Zy)=\begin{pmatrix}v\\0\\0\end{pmatrix},\qquad
 W(SXy)=\begin{pmatrix}0\\v\\0\end{pmatrix},\qquad
 W(TXy)=\begin{pmatrix}0\\0\\v\end{pmatrix}.
 \end{aligned}
\end{equation}
Indeed, these equalities follow respectively by taking
$A=E_{11},E_{21},E_{31}$ with $B=\binom{0}{v}$;
$A=E_{11}$ with $B=\binom{v}{0}$;
$A=E_{12}$ with $B=\binom{0}{v}$; and
$A=E_{21},E_{31}$ with $B=\binom{v}{0}$.
In particular,
\[
 \forall y\in J,\qquad Sy\in H,\quad Ty\in H,\quad Xy-Zy\in H.
\]
For $u_1,u_2,u_3\in J$, let $v_1,v_2,v_3$ be their unique
parameters in~\eqref{eq:J}. Equation~\eqref{eq:table} gives
\[
 W(Zu_1+SXu_2+TXu_3)=\begin{pmatrix}v_1\\v_2\\v_3\end{pmatrix}.
\]
Since $H=\ker W$ and the parametrization is injective,
\[
 \forall u_1,u_2,u_3\in J,\qquad
 Zu_1+SXu_2+TXu_3\in H\ \Longrightarrow\ u_1=u_2=u_3=0.
\]
All hypotheses of Lemma~\ref{lem:count} are therefore satisfied,
and
\[
 5(r-3m)=5\dim H\geq9m.
\]
Thus $r\geq24m/5$.

\textbf{Exclusion of equality.}
Suppose $r=24m/5$. Then $5\dim H=9m$, and
Lemma~\ref{lem:count} yields a subspace $F\subseteq J$ with
$F\neq\{0\}$, $S(F)\subseteq F$, and $T(F)\subseteq F$.
By Lemma~\ref{lem:support}, there exist
$y\in F\setminus\{0\}$ and $(s,t)\in\K^2$ such that
$Sy=sy$ and $Ty=ty$. Comparing coordinates gives
\[
 \forall i\in\{1,\ldots,r\},\qquad
 y_i\neq0\ \Longrightarrow\ (s_i,t_i)=(s,t).
\]
The normalization property in Lemma~\ref{lem:normalization} implies
\[
 \dim\operatorname{span}\{\alpha_i:y_i\neq0\}\leq2.
\]
Since $0\neq y\in J$, there exists a unique nonzero
$v\in\K^{1\times m}$ with
$y=(\beta_i(\binom{0}{v}))_i$. Fix $B=\binom{0}{v}$.
Then~\eqref{eq:algorithm} becomes
\begin{equation}\label{eq:contradiction}
 AB=\sum_{i:y_i\neq0}\alpha_i(A)y_i C_i,
 \qquad\forall A\in\K^{3\times2}.
\end{equation}
The linear map defined by the right-hand side has rank at most
two: expansion of the forms $\alpha_i$ in a basis of their span
expresses its image in the span of at most two fixed matrices.
On the other hand,
\[
 \{AB:A\in\K^{3\times2}\}=\{zv:z\in\K^3\}.
\]
The map $z\mapsto zv$ is injective, since $v\neq0$.
Thus the left-hand side of~\eqref{eq:contradiction}, as a linear
map of $A$, has rank three, a contradiction. It follows that
$r>24m/5$.
\end{proof}

\section{Automated verification}

The entire argument above has been checked mechanically in Lean~4 against mathlib (Lean \texttt{4.19.0}).

\bibliographystyle{abbrvnat}
\bibliography{main}

@article{yang2026newlowerboundstensor,
      title={New lower bounds on tensor rank of $(2,n,m)$ matrix multiplication with GPT-6}, 
      author={Jason Yang},
      year={2026},
      eprint={2609.14393},
      archivePrefix={arXiv},
      primaryClass={cs.CC},
      url={https://arxiv.org/abs/2609.14393}, 
      journal={arXiv preprint arXiv:2609.14393},
}

@article{hopcroft1971minimizing,
  title={On minimizing the number of multiplications necessary for matrix multiplication},
  author={Hopcroft, John E and Kerr, Leslie R},
  journal={SIAM Journal on Applied Mathematics},
  volume={20},
  number={1},
  pages={30--36},
  year={1971},
  publisher={SIAM}
}

@article{ja1985improved,
  title={Improved lower bounds for some matrix multiplication problems},
  author={Ja'Ja', Joseph and Takche, Jean},
  journal={Information processing letters},
  volume={21},
  number={3},
  pages={123--127},
  year={1985},
  publisher={Elsevier}
}

@article{blaser1999lower,
  title={Lower bounds for the multiplicative complexity of matrix multiplication},
  author={Bl{\"a}ser, Markus},
  journal={computational complexity},
  volume={8},
  number={3},
  pages={203--226},
  year={1999},
  publisher={Springer}
}

@misc{shitov2026,
    title={Determining the tensor rank of $(2, 2, M)$ matrix multiplication with GPT-5.6 Sol.}, 
      author={Yaroslav Shitov},
      year={2026},
      url={https://doi.org/10.13140/RG.2.2.29432.61441}, 
}
\end{document}